\documentclass[a4paper,twocolumn,10pt,accepted=2021-06-15]{quantumarticle}
\pdfoutput=1

\usepackage{amsmath, amssymb, amsthm}
\usepackage{mathtools}
\usepackage{algorithm, algpseudocode}
\usepackage{color}
\usepackage{float}
\usepackage{braket}
\usepackage[colorlinks, breaklinks]{hyperref}
\usepackage{enumitem}
\usepackage{ulem}
\usepackage{booktabs}
\usepackage[numbers,sort&compress]{natbib}

\DeclareMathOperator{\tr}{Tr}

\DeclareMathOperator{\id}{id}

\theoremstyle{definition}
\newtheorem{theorem}{Theorem}

\newtheorem{proposition}[theorem]{Proposition}
\theoremstyle{definition}
\newtheorem{definition}[theorem]{Definition}

\theoremstyle{remark}
\newtheorem{conjecture}{Conjecture}
\newtheorem{remark}[conjecture]{Remark}

\begin{document}
\author{Hayata Yamasaki}
\email{hayata.yamasaki@gmail.com}
\affiliation{Photon Science Center, Graduate School of Engineering, The University of Tokyo, 7--3--1 Hongo, Bunkyo-ku, Tokyo 113--8656, Japan}
\affiliation{Institute for Quantum Optics and Quantum Information (IQOQI), Austrian Academy of Sciences, Boltzmanngasse 3, 1090 Vienna, Austria}
\affiliation{Atominstitut,  Technische  Universit{\"a}t Wien, Stadionallee 2, 1020 Vienna, Austria}
\author{Madhav Krishnan Vijayan}
\email{madhavkrishnan.vijayan@uts.edu.au}
\affiliation{Centre for Quantum Software \& Information (UTS:QSI), University of Technology Sydney, Sydney NSW, Australia}
\author{Min-Hsiu Hsieh}
\email{min-hsiu.hsieh@foxconn.com}
\affiliation{Centre for Quantum Software \& Information (UTS:QSI), University of Technology Sydney, Sydney NSW, Australia}
\affiliation{Hon Hai Quantum Computing Research Center, Taipei City, Taiwan}

\title{Hierarchy of quantum operations in manipulating coherence and entanglement}

\begin{abstract}
Quantum resource theory under different classes of quantum operations advances multiperspective understandings of inherent quantum-mechanical properties, such as quantum coherence and quantum entanglement.
We establish hierarchies of different operations for manipulating coherence and entanglement in distributed settings, where at least one of the two spatially separated parties are restricted from generating coherence.
In these settings, we introduce new classes of operations and also characterize those maximal, \textit{i.e.}, the resource-non-generating operations,
progressing beyond existing studies on incoherent versions of local operations and classical communication and those of separable operations.
The maximal operations admit a semidefinite-programming formulation useful for numerical algorithms, whereas the existing operations not.
To establish the hierarchies, we prove a sequence of inclusion relations among the operations by clarifying tasks where separation of the operations appears.
We also demonstrate an asymptotically non-surviving separation of the operations in the hierarchy in terms of performance of the task of assisted coherence distillation, where a separation in a one-shot scenario vanishes in the asymptotic limit.
Our results serve as fundamental analytical and numerical tools to investigate interplay between coherence and entanglement under different operations in the resource theory.
\end{abstract}

\maketitle

\section{Introduction}

Advantages of quantum information processing over conventional classical information processing, whether in computation~\cite{W5,H3}, communication~\cite{Wehnereaam9288}, or cryptography~\cite{Pirandola2019}, arise from various inherent properties of quantum mechanics, such as quantum coherence and quantum entanglement.
Quantum resource theories~\cite{RevModPhys.91.025001,K1,kuroiwa2021consistent} have grown to be an important theoretical framework for quantitative analyses of such properties from operational perspectives using information processing tasks.
A resource theory is conventionally defined by specifying a class of allowed operations as \textit{free operations}.
One way to choose free operations may be to use practical or experimental restrictions.
For example, the resource theory of entanglement can be defined by considering a distributed setting for multiple parties with access only to local operations on each party's quantum system~\cite{H2,P1,E5}; then, local operations and classical communication (LOCC)~\cite{doi:10.1063/1.1495917,Chitambar2014,Y5} may arise as a natural candidate for free operations.
Entanglement serves as a resource for distributed quantum information processing where spatially separated parties are restricted to LOCC, by enabling quantum teleportation~\cite{B5} and allowing for the implementation of nonlocal operations to the shared system~\cite{doi:10.1002/qute.201800066,PhysRevA.96.032330,Y5}.
Yet importantly, to deepen our understandings of entanglement, it is also crucial to introduce and exploit larger classes of free operations than LOCC, such as separable (SEP) operations~\cite{PhysRevLett.78.2275,Rains1997} and positive-partial-transpose (PPT) operations~\cite{Rains2000}, in analytical and numerical studies of entanglement~\cite{H2,P1,E5}.
Along with the studies of entanglement, distributed settings also commonly arise in other resource theories, where each party has a restricted power of manipulating given quantum resources rather than performing arbitrary local operations, and needs assistance of another party in using the given resources~\cite{M1,R1,PhysRevLett.116.070402,PhysRevLett.122.130601,PhysRevA.101.052305}.

In this paper, we investigate the distributed settings that involve two prominent resource theories, entanglement and coherence~\cite{PhysRevLett.117.020402,S2,Matera_2016,PhysRevX.6.041028,PhysRevX.8.031005}.
In particular, in the spirit of studying LOCC, SEP, and PPT operations in entanglement theory,
we introduce and study different natural classes of operations in the distributed settings of manipulating coherence and entanglement, and compare their relative power in performing information theoretic tasks.
Coherence, \textit{i.e.}, superposition of a certain set of quantum states, has been shown to play important roles in quantum biology~\cite{doi:10.1080/00405000.2013.829687}, quantum thermodynamics~\cite{goold2016role} and photonic experiments~\cite{doi:10.1098/rspa.2017.0170}, where certain states are easier to create than their superposition.
The resource theory of coherence~\cite{S3} is a well-established resource theory that is useful for introducing classifications, partial orders, and quantifications of quantum coherence.
The resource theory of coherence considers situations where coherence cannot be created on a quantum system due to a restriction of operations for manipulating the system.
The free states in the resource theory of coherence are states represented as diagonal density operators in some fixed basis.
As is the case of entanglement, several different free operations that preserve diagonal density operators have been well investigated, such as incoherent operations (IO)~\cite{PhysRevLett.113.140401,PhysRevLett.116.120404}, maximally incoherent operations (MIO)~\cite{ABERG2004326,Aberg2006}, strictly incoherent operations (SIO)~\cite{PhysRevLett.116.120404,PhysRevX.6.041028}, and physically incoherent operations (PIO)~\cite{PhysRevLett.117.030401,PhysRevA.94.052336}, to name a few.
The resource theory of coherence in the \textit{distributed} settings has also attracted attentions of broad interests~\cite{PhysRevLett.117.020402,S2,Matera_2016,PhysRevX.6.041028,PhysRevX.8.031005}, as with the distributed settings in other resource theories.
These resource theories for the distributed manipulation of coherence provide a framework for investigating an interplay between coherence and entanglement in various information processing tasks such as distillation and dilution of these resources~\cite{PhysRevLett.117.020402}, assisted distillation of coherence~\cite{PhysRevLett.116.070402,M1,R1}, quantum state merging~\cite{PhysRevLett.116.240405}, quantum state redistribution~\cite{A11}, and multipartite state transformation~\cite{PhysRevA.99.042306,PhysRevA.100.052316}.
From a practical perspective, the distributed manipulation of coherence naturally arises in photonic systems as demonstrated in recent experiments~\cite{Wu:17,PhysRevLett.121.050401,npj_6_22}.

\begin{figure}[t]
  \centering
  \includegraphics[width=3.4in]{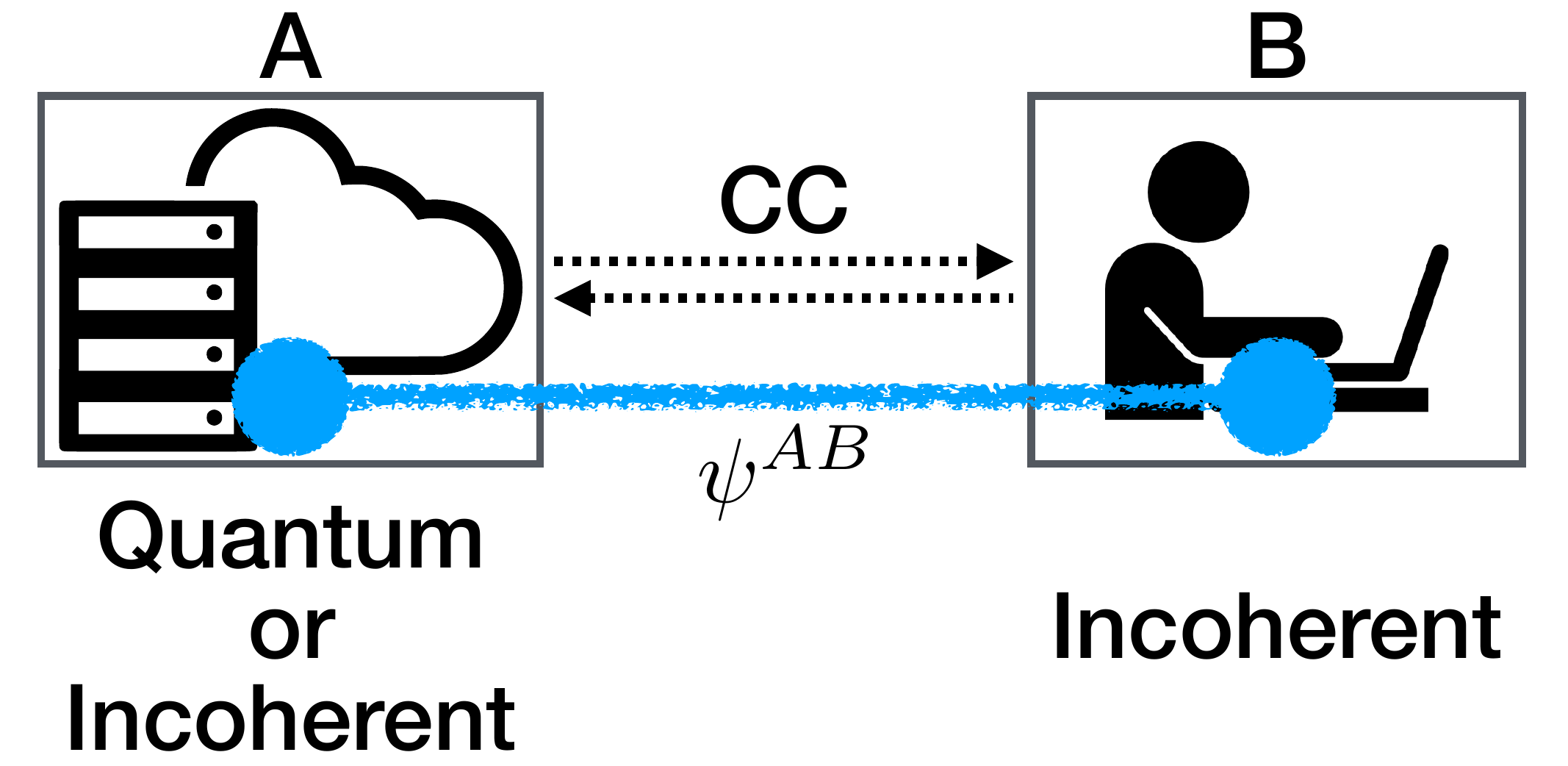}
  \caption{\label{fig:introduction}Distributed manipulation of coherence and entanglement of a quantum state $\psi^{AB}$ shared between two specially separated parties $A$ and $B$. The parties manipulate their local quantum systems, while they can use classical communication (CC). Local operations and classical communication with one party $B$ restricted to incoherent operations are called LQICC, which can be regarded as a client-server setting where the ability of the client $B$ to generate coherence is restricted while the server $A$ can perform any local quantum operation to assist $B$. Local operations and classical communication with both parties $A$ and $B$ restricted to incoherent operations are called LICC, where the abilities of $A$ and $B$ are the same.}
\end{figure}

In the distributed settings of manipulating coherence, especially for two parties $A$ and $B$, LOCC with one party restricted to IO are called LQICC, and those with both parties restricted to IO are LICC~\cite{S2}, as depicted in Fig.~\ref{fig:introduction}.
LQICC can be regarded as a client-server setting where the client's ability to generate coherence is restricted, while the abilities of two parties in LICC are the same.
The set of free states for LQICC, that is, the states that can be obtained from any initial state by LQICC, yields the set of quantum-incoherent (QI) states, which is incoherent on one of the parties.
The set of free states for LICC is incoherent states on both of the two parties.

However, to obtain deeper understandings,
it is crucial to introduce and compare different classes of operations beyond LQICC and LICC\@.
For example, in the entanglement theory, LOCC may be a conventional and well-motivated choice of free operations, but the set of LOCC is mathematically difficult to characterize and analyze.
To circumvent this problem, more general classes of operations than LOCC, such as separable operations and PPT operations, are vital to investigating performances of information processing tasks, which also yields bounds of the performance under LOCC\@.
Especially, PPT operations provide numerical algorithms for calculating performance of the entanglement-assisted tasks by means of semidefinite programming (SDP)~\cite{Watrous:2018:TQI:3240076}, even if the corresponding tasks under LOCC are hard to analyze due to its mathematical structure~\cite{Rains2000,PhysRevA.94.050301,PhysRevLett.119.180506,8707001,Regula_2019,Wang2018,PhysRevLett.125.040502}.
Similarly, in the resource theory of coherence, MIO serves as a class of operations beyond IO, and MIO provides numerical algorithms based on SDP similarly to PPT operations~\cite{PhysRevLett.121.010401}.
Importantly, even if the operations such as SEP, PPT, and MIO are defined mathematically, these different classes of operations provide efficiently calculable bounds in analyzing the information processing tasks and crucially help us to understand the properties of resources in the study of the resource theories.
In the same way, in our distributed settings of manipulating coherence and entanglement, we may suffer from the difficulty if the operations are restricted to LQICC and LICC\@.
To circumvent this difficulty, Ref.~\cite{S2} introduced a class of operations generalizing LQICC by considering SEP with one party restricted to IO, called SQI, and that generalizing LICC by considering SEP with both parties restricted to IO, called SI\@.
However, SQI and SI are insufficient for providing numerically tractable algorithms in these resource theories, in contrast with PPT operations for entanglement and MIO for coherence.

\begin{figure}[t]
  \centering
  \includegraphics[width=3.4in]{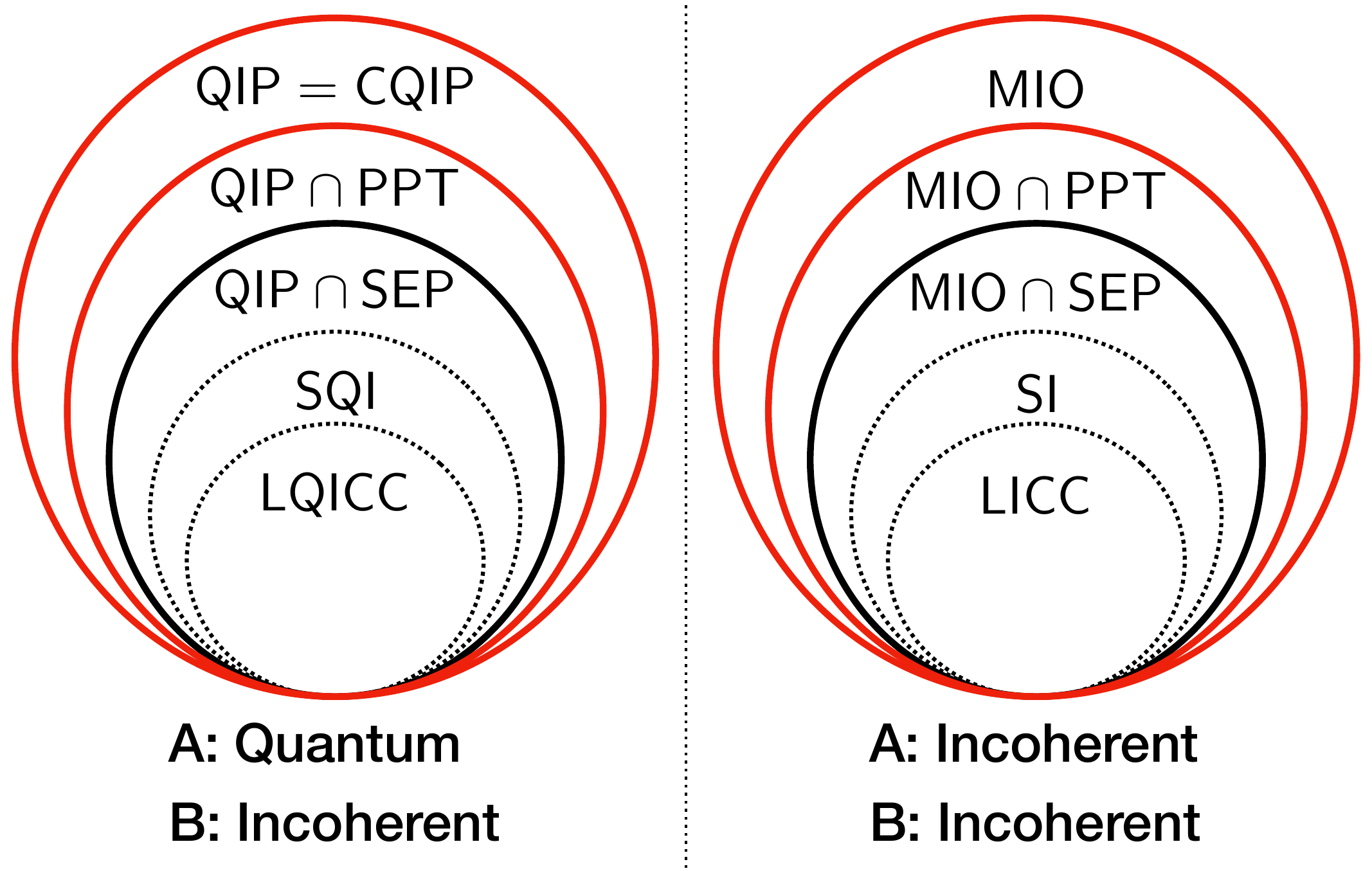}
  \caption{\label{fig:hierarchy}A hierarchy of operations for distributed manipulation of coherence and entanglement over two parties $A$ and $B$ beyond LQICC and SQI on the left, and that beyond LICC and SI on the right. The classes of operations that we introduce and characterize in this paper are shown by red and black bold circles, while the existing classes are black dotted circles. Especially, the red bold circles represent the classes of operations that have characterizations in terms of semidefinite programming (SDP), and hence can be used for numerical algorithms based on SDP\@.}
\end{figure}

Progressing beyond the above previous research on LQICC, LICC, SQI, and SI, we here introduce and analyze several classes of operations illustrated in Fig.~\ref{fig:hierarchy},
which are also summarized in Table~\ref{tab:result}.
As a generalization of LQICC and SQI to a class of operations that transform any QI state into a QI state~\cite{PhysRevLett.116.160407}, we consider the QI-preserving (QIP) operations.
As for the generalization of LICC and SI as a counterpart of QIP, we use MIO that preserves incoherent states on both of the two party.
We show that no inclusion relation holds between QIP and MIO on two parties, which contrasts with the fact that LQICC includes LICC and SQI includes SI\@.
The QIP is analogous to separability-preserving (SEPP) operations~\cite{5714245} in entanglement theory, which transforms any separable state into a separable state rather than QI states.
In entanglement theory, separable operations, which preserve the separability even if applied to a subsystem, are known to be a strict subclass of separability-preserving operations~\cite{Chitambar2017}.
In contrast, considering a seemingly smaller class of operations that is QI-preserving even if applied to a subsystem, which we define as completely QI-preserving (CQIP) operations, we prove that CQIP actually coincides with QIP exactly.
We introduce the new classes of operations shown in Fig.~\ref{fig:hierarchy} by taking intersection of different sets of operations.
These generalized classes of operations beyond LQICC, SQI, LICC, ans SI are advantageous in analyzing tasks because we can characterize some of them by SDP, as shown in Fig.~\ref{fig:hierarchy}, to provide SDP-based numerical algorithms for analyzing resource theories for distributed manipulation of coherence and entanglement.

Using these operations, we establish hierarchies (Fig.~\ref{fig:hierarchy}) of operations in the resource-theoretic framework of manipulating coherence and entanglement in the distributed settings.
In the entanglement theory, it has been a crucial question when the difference between LOCC, SEP, and PPT operations appears in the performance of achieving tasks, such as local state discrimination~\cite{B6,C4,6747300,B1,CH14,6687245} and entanglement manipulation~\cite{Rains2000,PhysRevA.94.050301,PhysRevLett.119.180506,8707001,Regula_2019,Wang2018,PhysRevLett.125.040502,Chitambar2017,PhysRevLett.121.190504}.
Also in the resource theory of coherence, the difference between MIO, IO, and other possible restricted operations has been investigated for tasks such as coherence manipulation~\cite{PhysRevLett.116.120404,PhysRevLett.121.010401,PhysRevA.97.050301,8863412}.
As for our case of distributed manipulation of coherence and entanglement, we identify tasks where differences of the operations in the hierarchy appear, leading to a proof of the strict separation of the newly introduced operations.

\begin{table*}
  \caption{\label{tab:result}Summary of definitions and examples of operations in resource theories for distributed manipulation of coherence and entanglement. The examples in the table lead to our results on the strict inclusion relations among the operations shown in Fig.~\ref{fig:hierarchy}.}
  \begin{tabular}{l@{\hspace{0.5cm}}p{4.5cm}@{\hspace{0.5cm}}p{8.5cm}}
    \toprule
    Set of operations & Definition & Examples that lead to our results on separation\\
    \midrule
    $\mathsf{QIP}$ & (Definition~\ref{def:mio}. See also~\eqref{eq:qi_state} for the definition of QI states shared between $A$ and $B$.) Operations (CPTP linear maps) on $A$ and $B$ that map any QI state shared between $A$ and $B$ to a QI state. & (Proposition~\ref{prp:fidelity}) An operation on $A$ and $B$ achieving the task of one-shot assisted distillation of coherence given in Proposition~\ref{prp:fidelity}.\\
    $\mathsf{QIP}\cap\mathsf{PPT}$ & Operations in the intersection of \textsf{QIP} and \textsf{PPT}. & (Proposition~\ref{prp:mio_ppt}) A map implemented by a PPT measurement with $B$ outputting a classical state representing the measurement outcome.\\
    $\mathsf{QIP}\cap\mathsf{SEP}$ & Operations in the intersection of \textsf{QIP} and \textsf{SEP}. & (Proposition~\ref{prp:sep}) MIO on $B$ with $A$ having a one-dimensional system.\\
    \textsf{SQI} & (Definition~\ref{def:sqi_si})  SEP where $B$'s Kraus operators satisfy the condition~\eqref{eq:io_kraus} of IO\@. &
    (Proof of Theorem~\ref{thm:inclusion}, Proposition~\ref{prp:qip_sep_mio_sep}) A map implemented by a separable measurement with $B$ outputting a classical state representing the measurement outcome.\\
    \textsf{LQICC} & (Definition~\ref{def:lqicc_licc})  LOCC with restriction on $B$'s operations to IO\@. & (Proof of Theorem~\ref{thm:inclusion}) A map implemented by an LOCC measurement with $B$ outputting a classical state representing the measurement outcome.

    (Proposition~\ref{prp:sep}) IO on $B$ with $A$ having a one-dimensional system.\\
    \midrule
    $\mathsf{MIO}$ & (Definition~\ref{def:mio}. See also~\eqref{eq:I} for the definition of incoherent states shared between $A$ and $B$.) Operations (CPTP linear maps) on $A$ and $B$ that map any incoherent state shared between $A$ and $B$ to an incoherent state. & (Proposition~\ref{prp:mio_ppt}) A SWAP unitary operation between $A$ and $B$.\\
    $\mathsf{MIO}\cap\mathsf{PPT}$ & Operations in the intersection of \textsf{MIO} and \textsf{PPT}. & (Proposition~\ref{prp:qip_sep_mio_sep}) A map implemented by a PPT measurement with both $A$ and $B$ outputting a classical state representing the measurement outcome.\\
    $\mathsf{MIO}\cap\mathsf{SEP}$ & Operations in the intersection of \textsf{MIO} and \textsf{SEP}. &
    (Proposition~\ref{prp:sep}) MIO on one of the parties $A$ and $B$ with the other party having a one-dimensional system.\\
    \textsf{SI} & (Definition~\ref{def:sqi_si}) SEP where both $A$ and $B$'s Kraus operators satisfy the condition~\eqref{eq:io_kraus} of IO\@. & (Proof of Theorem~\ref{thm:inclusion}, Proposition~\ref{prp:qip_sep_mio_sep}) A map implemented by a separable measurement with both $A$ and $B$ outputting a classical state representing the measurement outcome.\\
    \textsf{LICC} & (Definition~\ref{def:lqicc_licc}) LOCC with restriction on $A$ and $B$'s operations to IO\@. & (Proof of Theorem~\ref{thm:inclusion}) A map implemented by an LOCC measurement with both $A$ and $B$ outputting a classical state representing the measurement outcome.

    (Proposition~\ref{prp:sep}) IO on one of the parties $A$ and $B$ with the other party having a one-dimensional system.\\
    \bottomrule
  \end{tabular}
\end{table*}

In contrast with these separations of the operations in the hierarchies, we also demonstrate an asymptotically non-surviving separation of these operations in terms of performance of a task.
In the context of when the difference between different classes of operations arises,
the asymptotically non-surviving separation refers to a phenomena that the difference in the performance of a task in a one-shot scenario disappears in the corresponding asymptotic scenario.
An asymptotically non-surviving separation has been recently discovered in Ref.~\cite{Y1} in a communication task of quantum state merging~\cite{Horodecki2005,Horodecki2006,8590809}.
Reference~\cite{Y1} discloses this phenomena by proving the difference in the required amount of entanglement resources in achieving one-shot quantum state merging under one-way and two-way LOCC, while this difference ceases to exist in the conventional asymptotic quantum state merging.
In contrast with this communication task, we here consider another resource-theoretic task that is known as assisted distillation of coherence~\cite{PhysRevLett.116.070402,M1,R1}.
We demonstrate the difference in the maximal amount of distillable coherence with assistance under QIP and one-way LQICC in a one-shot scenario, using the SDP calculation that we establish in this paper, and at the same time prove the coincidence of them in the corresponding asymptotic scenario.
This demonstration finds an application of the SDP-based numerical technique that we introduce for investigating manipulation of coherence and entanglement, which by itself has wide applicability beyond the assisted distillation of coherence.

Consequently, our results establish essential analytical and numerical tools for investigating resource theories for manipulating coherence and entanglement in the distributed settings, by introducing different classes of operations and clarifying relative powers of the operations.
The significance of our results is that the operations that we prove to have different powers are useful for bounding performances of information processing tasks performed by the distributed manipulation of coherence and entanglement.
The results are also of fundamental importance in clarifying similarity and difference between the manipulation of entanglement in entanglement theory and the distributed manipulation of coherence in this paper. This will be useful for better understanding the interplay between coherence and entanglement from the operational viewpoint, based on resource theories under different operations.

The rest of this paper is organized as follows.
In Sec.~\ref{sec:preliminaries}, we recall operations that are known in resource theories of entanglement, coherence, and distributed manipulation of coherence and entanglement.
In Sec.~\ref{sec:qip}, we introduce QI-preserving operations and prove the equivalence between QI-preserving operations and completely QI-preserving operations.
In Sec.~\ref{sec:hierarchy}, we establish a hierarchy of operations beyond LQICC and SQI and that beyond LICC and SI, clarifying tasks that separate the power of these operations.
In Sec.~\ref{sec:non_surviving_separation}, we demonstrate the asymptotically non-surviving separation between LQICC and QI-preserving operations.
Our conclusion is given in Sec.~\ref{sec:conclusion}.

\section{\label{sec:preliminaries}Preliminaries}

As we are interested in manipulation of coherence and entanglement in distributed settings, the first question to ask is what are the allowed operations in such settings. The choices of allowed operations {for spatially separated parties $A$ and $B$, \textit{e.g.},}
local operations with classical communication ($\mathsf{LOCC}$), separable operations $(\mathsf{SEP})$ and positive partial transpose maps $(\mathsf{PPT})$, have been well explored in entanglement theory~\cite{H2}. Similarly there are several candidates for allowed operations that do not create coherence in a fixed reference basis~\cite{S3}. We will focus our attention on two prominent ones, namely, incoherent operations (\textsf{IO}) and maximally incoherent operations (\textsf{MIO}). We will now define these operations formally and discuss how they can be combined to investigate distributed manipulation of coherence and entanglement.

In this paper, we consider quantum systems represented as finite-dimensional Hilbert spaces, which are written as $\mathcal{H}^A\otimes\mathcal{H}^{A^\prime}\otimes\cdots$ on $A$ and $\mathcal{H}^B\otimes\mathcal{H}^{B^\prime}\otimes\cdots$ on $B$.
We fix an orthonormal basis on $A$ and $B$, which is called a reference basis and denoted by ${\{\Ket{i}^A\}}_{i=0,\ldots,D_A}$ of $\mathcal{H}^A$ and ${\{\Ket{j}^B\}}_{j=0,\ldots,D_B}$ of $\mathcal{H}^B$, where we write the dimension of $\mathcal{H}^A$ and $\mathcal{H}^B$ as $D_A$ and $D_B$, respectively.
We may use superscripts of bras and kets to represent the quantum system to which the bras and kets belong, such as $\Ket{\psi}^{AB}\in\mathcal{H}^A\otimes\mathcal{H}^B$.
Superscripts of operators and maps represent the quantum system on which they act, such as $\mathcal{E}^{AB}$ for a linear map of operators on $\mathcal{H}^A\otimes\mathcal{H}^B$ and $\mathcal{E}^{A\to A^\prime}$ for that from $\mathcal{H}^A$ to $\mathcal{H}^{A^\prime}$.
The identity operator is denoted by $\openone$, and the identity map is $\id$.

\subsection{Operations in resource theories of entanglement}

We review here definitions of operations that are important in the resource theory of entanglement~\cite{H2}.

We recall the definition of LOCC as follows.

\begin{definition}[\textsf{LOCC}]
  Local operations with classical communication or \textsf{LOCC} can be defined for a bipartite system as follows: if a completely positive and trace-preserving (CPTP) map $\mathcal E^{AB}$ has a Kraus decomposition using local operations of measurements by parties $A$ and $B$ along with classical communication of the measurement outcomes between them, then the map represents \textsf{LOCC}\@. The set of $r$-round \textsf{LOCC} maps, denoted as $r$-\textsf{LOCC}, consists of an $\mathsf{LOCC}$ map that can be implemented using $r$ rounds of classical communication, where the case of $r=0$ yields local operations (without classical communication). Note that it can be shown that $r$-\textsf{LOCC} is a strict subset of $(r + 1)$-\textsf{LOCC}~\cite{Chitambar2014}. {\textsf{LOCC} can be thought to be a \textit{physical} operation in the resource theory of entanglement in the sense that it describes all the allowed space-like separated operations that the two classically communicating parties can perform.} For a more rigorous definition, see Refs.~\cite{doi:10.1063/1.1495917,Chitambar2014,Y5}.
\end{definition}

Separable operations include LOCC as special cases.
While there can be separable operations that are not LOCC, separable operations are often used as a relaxation of LOCC since separable operations can be mathematically easier to characterize than LOCC itself.

\begin{definition}[\textsf{SEP}]
  A CPTP map $\mathcal{E}^{AB}$ represents separable operations if and only if there exists a Kraus decomposition of $\mathcal{E}^{AB}$ where Kraus operators are product operators, \textit{i.e.},
  \begin{align}
      &\mathcal{E}^{AB}\left(\rho^{AB}\right)=\sum_i \left(M_i^A\otimes K_i^B\right)\rho^{AB} {\left(M_i^A\otimes K_i^B\right)}^\dag,\\
      &\sum_i {\left(M_i^A\otimes K_i^B\right)}^{\dagger}\left(M_i^A\otimes K_i^B\right) = \openone^{AB},
  \end{align}
 where $\openone^{AB}$ is the identity operator on the full space of A and B.
 The set of separable operations is denoted by  \textsf{SEP}\@.
 The separable measurement refers to a measurement represented by the measurement operators ${\left\{M_i^A\otimes K_i^B\right\}}_i$ satisfying the above condition of the separable maps.
 The set \textsf{LOCC} is a strict subset of the set of \textsf{SEP} operations~\cite{Chitambar2014}.
\end{definition}

PPT operations include separable operations and LOCC as special cases.
PPT operations may generate PPT entangled states from separable states, while separable operations and LOCC do not.
However, PPT operations are useful for analyzing tasks since the condition~\eqref{eq:ppt_sdp} in the following can be used for numerical algorithms based on semidefinite programming, unlike separable operations and LOCC\@.

\begin{definition}[\textsf{PPT}]
  A CPTP map $\mathcal{E}^{AB}$ represents a \textsf{PPT} operation if and only if the map
  \begin{equation}
    \left(\id^A\otimes T^B\right)\circ\mathcal{E}^{AB}\circ\left(\id^A\otimes T^B\right)
  \end{equation}
  is completely positive, where $\id^A$ is the identity map on $A$, and $T^B$ is the transpose map on $B$ with respect to the reference basis.
  Or equivalently, a CPTP map $\mathcal{E}^{AB}$ is a \textsf{PPT} map if and only if the Choi operator ${J\left(\mathcal{E}\right)}^{A^\prime B^\prime A B}$~\cite{Watrous:2018:TQI:3240076} satisfies a condition~\cite{Rains2000}
  \begin{equation}
    \label{eq:ppt_sdp}
      {J\left(\mathcal{E}\right)}^{\mathrm{T}_{A^\prime A}}\geqq 0,
  \end{equation}
  where $\mathrm{T}_{A^\prime A}$ is the partial transpose on $A^\prime A$ with respect to the reference basis.
  A PPT measurement~\cite{PhysRevLett.109.020506,C4,6747300,B1} can be represented by a family of measurement operators ${\left\{K_i^{AB}\right\}}_i$ satisfying
  \begin{align}
      \sum_i K_i^\dag{}^{AB} K_i^{AB}&=\openone,\\
       {\left(K_i^\dag{}^{AB} K_i^{AB}\right)}^{\mathrm{T}_A}&\geqq 0,\,\forall i.
  \end{align}
  Note that for any PPT measurement given by  ${\left\{K_i^{AB}\right\}}_i$, 
  a CPTP map implemented by this measurement
  \begin{equation}
  \mathcal{E}^{AB}\left(\rho^{AB}\right)=\sum_i K_i^{AB}\rho^{AB} K_i^\dag{}^{AB}
  \end{equation}
  is a \textsf{PPT} map.
  The set \textsf{SEP} is a strict subset of \textsf{PPT}~\cite{Chitambar2014}.
\end{definition}

\subsection{Operations in resource theories of coherence}

Free operations in the resource theory of coherence are ones that map diagonal states to diagonal states in a particular reference basis~\cite{S3},
so that the diagonal states can be free states for the free operations in the resource theory of coherence.
Recall the reference bases  ${\{\Ket{i}^A\}}_i$ of $\mathcal{H}^A$ and ${\{\Ket{j}^B\}}_j$ of $\mathcal{H}^B$.
Then, the reference basis of a bipartite system $\mathcal{H}^A\otimes\mathcal{H}^B$ is the product reference bases of the subsystems.
The set of incoherent states of $\mathcal{H}^A$ is given by
\begin{equation}
  \left\{\sum_i p\left(i\right)\Ket{i}\Bra{i}^A\right\},
\end{equation}
where $p$ is a probability distribution, and the set of incoherent states of $\mathcal{H}^B$ can be given in the same way.
We recall that the class of incoherent operations \textsf{IO} refers to those which cannot create coherent states from incoherent states even with post-selection.
\begin{definition}[\textsf{IO}]
  A CPTP map $\mathcal{E}$ represents an incoherent operation if and only if there exists a Kraus decomposition satisfying
  \begin{equation}
    \label{eq:io}
    \begin{split}
      &\mathcal{E}\left(\rho\right)\coloneqq\sum_i K_i \rho K_i^\dag,\\
      &\forall i,\,\rho\in\mathcal{I}\Rightarrow \frac{K_i \rho K_i^\dag}{p\left(i\right)}\in\mathcal{I},\,p\left(i\right)\coloneqq \tr K_i \rho K_i^\dag,
    \end{split}
  \end{equation}
  where $\mathcal{I}$ denotes the set of incoherent states, and the trace-preserving condition yields $\sum_i K_i^\dag{}^B K_i^B=\openone$.
  Equivalently, these Kraus operators ${\left\{K_i^B\right\}}_i$ for incoherent operations can be written in the form of
  \begin{equation}
    \label{eq:io_kraus}
    K_i^B\coloneqq \sum_j c_i\left(j\right)\Ket{f_i\left(j\right)}\Bra{j}^B,\,\forall i,
  \end{equation}
  where $f_i$ is a function (which is not necessarily bijective), and $c_i\left(j\right)\in\mathbb{C}$ for each $j$~\cite{PhysRevLett.113.140401}.
  The set of incoherent operations is denoted by \textsf{IO}.
\end{definition}

While \textsf{IO} clearly preserves diagonal states in the reference basis, it is not the maximal set of CPTP maps that do so.
In a quantum resource theory, the maximal set of operations refers to the largest possible set of operations that do not generate any resource from free states~\cite{RevModPhys.91.025001}.
We now define this maximal set, {that is, maximally incoherent operations \textsf{MIO}.
Notice that unlike \textsf{IO}, \textsf{MIO} may create coherence from an incoherent state probabilistically by post-selecting a measurement outcome corresponding to a Kraus operator that implements \textsf{MIO}.}

\begin{definition}[\label{def:mio}\textsf{MIO}]
  A CPTP map $\mathcal{E}$ represents a maximally incoherent operation (\textsf{MIO}) if and only if $\mathcal{E}$ maps any incoherent state into an incoherent state deterministically~\cite{Aberg2006}, that is, for any $\rho\in\mathcal{I}$,
  \begin{equation}
    \mathcal{E}\left(\rho\right)\in\mathcal{I},
  \end{equation}
  where $\mathcal{I}$ denotes the set of incoherent states.
\end{definition}

\subsection{Operations in resource theories of distributed manipulation of coherence and entanglement}

In the two-party distributed settings for manipulating coherence and entanglement, we can look at two possible settings, namely, (i) a quantum-incoherent setting where one party is restricted to use incoherent operations, and (ii) an incoherent-incoherent setting where both parties are restricted to incoherent operations,
as depicted in Fig.~\ref{fig:introduction}.
We now recall the extensions of \textsf{LOCC} to the quantum-incoherent and incoherent-incoherent settings.

\begin{definition}[\label{def:lqicc_licc}\textsf{LQICC} and \textsf{LICC}]
  The set of \textsf{LQICC}~\cite{S2} is defined in the same way as \textsf{LOCC}, with restriction on $B$'s operations to incoherent operations. The set of \textsf{LICC}~\cite{S2} is also defined in the same way as \textsf{LOCC}, with restriction on both $A$ and $B$'s operations to incoherent operations. Corresponding to the $r$-round LOCC, we define $r$-\textsf{LQICC} and $r$-\textsf{LICC} by replacing \textsf{LOCC} with $r$-\textsf{LOCC} in the above definitions, respectively.
  One-way LQICC refers to operations represented as CPTP maps by $A$ and $B$ consisting of $A$'s (arbitrary quantum) measurement represented by a quantum instrument ${\left\{\mathcal{E}_i^A\right\}}_i$ with the outcome $i$, followed by $B$'s incoherent operation $\tilde{\mathcal{E}}_i^{B}\in\textsf{IO}$ conditioned on $i$,
  that is, CPTP maps in the form of
  \begin{equation}
    \label{eq:lqicc}
    \mathcal{E}^{\textup{1-LQICC}}\coloneqq\sum_i\mathcal{E}_i^A\otimes\tilde{\mathcal{E}}_i^{B},
  \end{equation}
  which we may also write as $1$-\textsf{LQICC} because one-way LQICC is composed of one round of classical communication.
  One-way LICC is defined in a similar way to the the definition of one-way LQICC using incoherent operations on both $A$ and $B$.
  Note that we always consider one-way classical communication from $A$ to $B$ in one-way LQICC\@.
\end{definition}

The set of free states for \textsf{LQICC} in the quantum-incoherent setting are quantum-incoherent (QI) states~\cite{PhysRevLett.116.070402}.
A bipartite state $\rho^{AB}$ is called a QI state if and only if $\rho^{AB}$ is in the form of
\begin{equation}
  \label{eq:qi_state}
  \rho^{AB}=\sum_j p\left(j\right)\rho_j^A\otimes{\Ket{j}\Bra{j}}^B,
\end{equation}
where $\{\Ket{j}^B\}$ is the incoherent basis of $\mathcal{H}^B$, $p\left(j\right)$ is a probability distribution, and $\rho_j^A\in\mathcal{D}\left(\mathcal{H}^A\right)$.
The set of QI states is denoted by $\mathcal{QI}$.
Note that QI states are different from, but a special case of, quantum-classical (QC) states, \textit{i.e.}, states that have zero quantum discord~\cite{RevModPhys.84.1655}, since the reference basis of $B$'s diagonal part is fixed for QI states but not for QC states.
The set $\mathcal{I}$ of incoherent states of a bipartite quantum system $AB$ refers to that of states which are incoherent on both $A$ and $B$, \textit{i.e.}
\begin{equation}
\label{eq:I}
\mathcal{I}=\left\{\sum_{i,j} p\left(i,j\right)\Ket{i}\Bra{i}^A\otimes\Ket{j}\Bra{j}^B\right\},
\end{equation}
where $p\left(i,j\right)$ is a probability distribution.
These are the free states for \textsf{LICC} in the incoherent-incoherent setting.

We now define a larger set of operations by considering separable operations in conjunction with incoherent operations.

\begin{definition}[\label{def:sqi_si}\textsf{SQI} and \textsf{SI}]
In analogy to separable maps, define SQI~\cite{S2} to be a class of CPTP maps by $A$ and $B$ as
\begin{equation}
  \label{eq:qisep}
  \mathcal{E}^{\textup{SQI}}\left(\rho^{AB}\right)\coloneqq\sum_i \left(M_i^A\otimes K_i^{B}\right)\rho^{AB}{\left(M_i^A\otimes K_i^{B}\right)}^\dag,
\end{equation}
where ${\left\{M_i^A\otimes K_i^B\right\}}_i$ is a family of Kraus operators satisfying
\begin{equation}
  \sum_i {\left(M_i^A\otimes K_i^B\right)}^\dag\left(M_i^A\otimes K_i^B\right)=\openone,
\end{equation}
and for each $i$, $K_i^B$ satisfies the condition~\eqref{eq:io_kraus} of incoherent operations.
The class \textsf{SI}~\cite{S2} is defined in the same way as \textsf{SQI} where not only $K_i^B$ but both $M_i^A$ and $K_i^B$ satisfy the condition~\eqref{eq:io_kraus} of incoherent operations.
\end{definition}

While these operations combine operations in the resource theory of entanglement and that of coherence, they are not necessarily equivalent to the intersection of the sets of two different operations in these resource theories.
We remark the subtleties of taking the intersection of two operations in the following.

\begin{remark}
  [\label{remark:op_int}Difference between the set intersection of CPTP maps and the operational intersection]
  It is not straightforward to consider intersection of two different sets of operations because of the fact that a CPTP map may have different Kraus operators ${\left\{M_i\right\}}_i$ and ${\left\{K_j\right\}}_j$ that are related by a unitary operator $u={\left(u_{i,j}\right)}_{i,j}$ as~\cite{Watrous:2018:TQI:3240076}
  \begin{equation}
    K_j=\sum_i u_{i,j}M_i.
  \end{equation}
  For example, consider a CPTP map
  \begin{align}
    &\mathcal{E}^{AB}\left(\rho^{AB}\right)\coloneqq\sum_{i=0}^{1}M_i\rho M_i^\dag,
  \end{align}
  where
  \begin{align}
    &M_0\coloneqq\Ket{+}\Bra{+}^A\otimes\openone^B,\\
    &M_1\coloneqq\Ket{-}\Bra{-}^A\otimes Z^B,
  \end{align}
  which is a one-way LOCC map that can be implemented by $A$'s measurement represented by Kraus operators as $\left\{\Ket{+}\Bra{+},\Ket{-}\Bra{-}\right\}$, followed by $B$'s correction $\openone^B$ or $Z^B$ conditioned on $A$'s outcome.
  At the same time, this CPTP map $\mathcal{E}^{AB}$ is in \textsf{IO} on $AB$ because $\mathcal{E}^{AB}$ can also be regarded using a different set of Kraus operators as
  \begin{align}
    &\mathcal{E}^{AB}\left(\rho^{AB}\right)\coloneqq\sum_{i=0}^{1}K_i\rho K_i^\dag,\\
    &K_0\coloneqq
    \frac{\openone^A\otimes\Ket{0}\Bra{0}^B+X^A\otimes\Ket{1}\Bra{1}^B}{\sqrt{2}},\\
    &K_1\coloneqq
    \frac{X^A\otimes\Ket{0}\Bra{0}^B+\openone^A\otimes\Ket{1}\Bra{1}^B}{\sqrt{2}},
  \end{align}
  which is an incoherent map that can be implemented by a \textit{nonlocal} incoherent measurement on $AB$.
  However, it is unclear whether $\mathcal{E}$ is in \textsf{LICC} in the sense that there exists an implementation by \textit{local} incoherent operations and classical communication.
  It is not straightforward to present Kraus operators of $\mathcal{E}$ that simultaneously satisfy the conditions of \textsf{LOCC} and \textsf{IO}.
  We propose that a sensible way to resolve this is to consider a subset of the intersection of two sets of CPTP maps, so that CPTP maps in this subset can be implemented by Kraus operators satisfying conditions of both of the original two simultaneously; \textit{e.g.}, for \textsf{LOCC} and \textsf{IO}, the Kraus operators are \textit{both} local and incoherent, which we name the \textit{operational intersection} of \textsf{LOCC} and \textsf{IO}.
\end{remark}

\section{\label{sec:qip}QI-preserving maps}

In this section, we investigate a class of operations that preserve QI states, which we name QI-preserving operations, or QIP for short. This QIP is a generalization of LQICC and SQI as we will investigate further in the next section.
As for a generalization of LICC and SI, we use MIO that preserves incoherent state shared between $A$ and $B$.
We show that no inclusion relation holds between QIP and MIO\@.

The QIP is analogous to the separability-preserving operations (SEPP)~\cite{5714245} in the entanglement theory, \textit{i.e.}, operations that preserve separable states.
It is known that SEPP is different from separable (SEP) operations, because the SWAP operation on a bipartite system preserves separability of the bipartite state, but the SWAP is not separable~\cite{Chitambar2017}.
In the context of the separability-preserving property, SEP is characterized as a strict subclass of separability-preserving operations because SEP preserves separability even if applied to a subsystem; that is,
a CPTP map $\mathcal{E}^{AB}$ is separable if and only if for any identity map $\id^{A^\prime B^\prime}$ on a bipartite auxiliary system $\mathcal{H}^{A^\prime}\otimes\mathcal{H}^{B^\prime}$, $\id^{A^\prime B^\prime}\otimes\mathcal{E}^{AB}$ is separability-preserving.
Similarly, a completely positive map is also characterized as a special subclass of positive maps that preserve positivity even if applied to a subsystem~\cite{Watrous:2018:TQI:3240076}.
Based on these observations, we can consider a seemingly special class of QIP that preserves QI states even if applied to a subsystem, which we call completely QI-preserving (CQIP) operations.
However, we show that QIP and CQIP are the same set of operations, which leads to an essential difference between the entanglement theory and the resource theory of distributed manipulation of coherence and entanglement.
In addition, the techniques for showing the equivalence of QIP and CQIP yield a characterization of QIP that can be used for numerical algorithms based on semidefinite programming (SDP).

We define QI-preserving operations as a class of bipartite operations that transforms any QI state into a QI state, which generalizes LQICC and SQI over the two parties\@.
Note that this class of operations is first considered in Ref.~\cite{PhysRevLett.116.160407} in the context of studying discord and coherence, but our key contributions in this paper are to clarify the relation between QIP and the other possible classes of operations in the resource theory of manipulating coherence and entanglement, and to provide the SDP-based characterization of QIP\@.
The formal definition of QIP is as follows.

\begin{definition}[\label{def:qip}\textsf{QIP}]
A CPTP map $\mathcal{E}^{AB}$ is said to be QI-preserving if and only if for any QI state $\rho^{AB}\in\mathcal{QI}$,
\begin{equation}
  \label{eq:qi_preserving}
  \mathcal{E}^{AB}\left(\rho^{AB}\right)\in\mathcal{QI}.
\end{equation}
The set of QI-preserving maps is denoted by $\mathsf{QIP}$.
\end{definition}

We can regard MIO over the two parties as a generalization of LICC and SI\@.
In particular, MIO in the following refers to the class of operations in Definition~\ref{def:mio} with the set of incoherent states $\mathcal{I}$ given by~\eqref{eq:I}, which is incoherent on both $A$ and $B$.
We have, by definition, an inclusion relation between LQICC and LICC, and the same inclusion between SQI and SI, that is,
\begin{align}
\label{eq:inclusion_relation}
  \mathsf{LICC}&\subsetneqq\mathsf{LQICC},\nonumber\\  
  \mathsf{SI}&\subsetneqq\mathsf{SQI}.
\end{align}
In contrast, QIP and MIO do not have this inclusion relation as we show in the following, which illuminates the difference between QIP and these previously known classes of operations.

\begin{proposition}
It holds that
\begin{align}
\label{eq:qpi_not_subset_mio}
\mathsf{QIP}&\not\subset\mathsf{MIO},\\
\label{eq:mio_not_subset_qip}
\mathsf{MIO}&\not\subset\mathsf{QIP}.
\end{align}
\end{proposition}

\begin{proof}
We prove~\eqref{eq:qpi_not_subset_mio} and~\eqref{eq:mio_not_subset_qip} by showing instances.

We have~\eqref{eq:qpi_not_subset_mio} by definition. In particular, a QI-preserving map can transform  an incoherent state $\Ket{0}^A\otimes\Ket{0}^B$ into a QI state $\Ket{+}^A\otimes\Ket{0}^B$, which is not incoherent on $A$, and hence, this map is not maximally incoherent in the sense that an incoherent state on the two parties is transformed into a state that is not in $\mathcal{I}$ defined as~\eqref{eq:I}.

To show~\eqref{eq:mio_not_subset_qip}, we construct a nonlocal IO that is not QI-preserving, while IO is included in MIO by definition.
Consider the following IO
\begin{align}
  \mathcal{E}^{AB}\left(\rho^{AB}\right)&=\sum_{i=0}^{1}K_i^{AB}\rho^{AB}K_i^\dag{}^{AB},
\end{align}
where
\begin{align}
    K_0^{AB}&=\Ket{0}\Bra{0}^A\otimes\Ket{0}\Bra{0}^B+\Ket{1}\Bra{1}^A\otimes\Ket{1}\Bra{0}^B,\\
    K_1^{AB}&=\Ket{0}\Bra{0}^A\otimes\Ket{0}\Bra{1}^B+\Ket{1}\Bra{1}^A\otimes\Ket{1}\Bra{1}^B.
\end{align}
This is IO since $\Ket{0}^A\otimes\Ket{0}^B$, $\Ket{0}^A\otimes\Ket{1}^B$, $\Ket{1}^A\otimes\Ket{0}^B$, and $\Ket{1}^A\otimes\Ket{1}^B$ are transformed into incoherent states by these Kraus operators $\left\{K_0,K_1\right\}$.
However, inputting a QI state $\Ket{+}^A\otimes\Ket{0}^B$ to this IO\@, we obtain an entangled state, that is, a non-QI state
\begin{equation}
    \mathcal{E}^{AB}\left(\Ket{+}\Bra{+}^A\otimes\Ket{0}\Bra{0}^B\right)=\Ket{\Phi^+}\Bra{\Phi^+}^{AB},
\end{equation}
where
\begin{equation}
    \Ket{\Phi^+}=\frac{1}{\sqrt{2}}\left(\Ket{0}^A\otimes\Ket{0}^B+\Ket{1}^A\otimes\Ket{1}^B\right).
\end{equation}
Thus, $\mathcal{E}^{AB}$ is not QI-preserving.
\end{proof}

We also define completely QI-preserving operations as a class of operations over two parties that is QI-preserving even if the operation is applied to subsystems shared between the two parties.
The formal definition is as follows.

\begin{definition}[\label{eq:def_cqip}\textsf{CQIP}]
A CPTP map $\mathcal{E}^{AB}$ is said to be completely QI-preserving if and only if for an identity map $\id^{A'B'}$ on any shared auxiliary system $\mathcal{H}^{A'}\otimes\mathcal{H}^{B'}$, the map $\id^{A'B'}\otimes\mathcal{E}^{AB}$ is QI-preserving. The set of completely QI-preserving maps is denoted by $\mathsf{CQIP}$.
\end{definition}

To show the equivalence of \textsf{QIP} and \textsf{CQIP},
we characterize QI-preserving operations using a finite set of equations that can also be used for numerical algorithms based on SDP\@.
Define a density operator of a $D_A$-dimensional system $A$ for any $a,b\in\left\{0,\ldots,D_A-1\right\}$
\begin{equation}
  \label{eq:basis_state}
  \rho_{a,b}^A\coloneqq \begin{cases}
    \Ket{a}\Bra{a}^A&\text{if }a=b,\\
    \frac{1}{2}{\left(\Ket{a}+\Ket{b}\right)\left(\Bra{a}+\Bra{b}\right)}^A&\text{if }a<b,\\
    \frac{1}{2}{\left(\Ket{a}+\mathrm{i}\Ket{b}\right)\left(\Bra{a}-\mathrm{i}\Bra{b}\right)}^A&\text{if }a>b.
  \end{cases}
\end{equation}
These density operators ${\left\{\rho_{a,b}^A\right\}}_{a,b}$ serve as a basis spanning the set of all the operators of this $D_A$-dimensional system $A$~\cite{Watrous:2018:TQI:3240076}.
Then, we show the following characterization of QI-preserving maps.
Note that it is straightforward to use the condition~\eqref{eq:qi_preserving_sdp} in the following proposition as a linear constraint of Choi operators of QI-preserving operations in SDP, as we will demonstrate in the proof of Proposition~\ref{prp:fidelity}.

\begin{proposition}[\label{prp:qi_preserving}Characterization of QI-preserving operations]
  A CPTP map $\mathcal{E}^{AB}$ is QI-preserving if and only if
  for any $a,b\in\left\{0,\ldots,D_A-1\right\}$ and any $j\in\left\{0,\ldots,D_B-1\right\}$,
  it holds that
  \begin{equation}
    \label{eq:qi_preserving_sdp}
    \mathcal{E}^{AB}\left(\rho_{a,b}^A\otimes\Ket{j}\Bra{j}^B\right)\in\mathcal{QI},
  \end{equation}
  where $\rho_{a,b}$ is defined as~\eqref{eq:basis_state}.
\end{proposition}

\begin{proof}
  Any QI-preserving map satisfies the condition~\eqref{eq:qi_preserving} by definition, and hence it suffices to show the converse; that is, we prove that any CPTP map satisfying~\eqref{eq:qi_preserving} transforms any QI state into a QI state.
  Consider any QI state
  \begin{equation}
    \label{eq:rho_ab_qip}
    \rho^{AB}=\sum_j p\left(j\right)\rho_j^A\otimes\Ket{j}\Bra{j}^B.
  \end{equation}
  Using density operators~\eqref{eq:basis_state} that serve as a basis of $\mathcal{H}^A$, we can write $\rho_j^A$ for each $j$ in~\eqref{eq:rho_ab_qip} as
  \begin{equation}
    \rho_j^A=\sum_{a,b}c_{a,b}^{\left(j\right)}\rho_{a,b}^A,
  \end{equation}
  where $c_{a,b}^{\left(j\right)}\in\mathbb{C}$ for each $a,b$ is a coefficient in this linear expansion of $\rho_j^A$.
  Then, due to the linearity of $\mathcal{E}^{AB}$, we have for each $j$
  \begin{align}
    &\mathcal{E}^{AB}\left(\rho_j^A\otimes\Ket{j}\Bra{j}^B\right)\\
    &=\sum_{a,b}c_{a,b}^{\left(j\right)}\mathcal{E}^{AB}\left(\rho_{a,b}^A\otimes\Ket{j}\Bra{j}^B\right)\in\mathcal{QI},
  \end{align}
  where the last inclusion follows from the convexity of the set of QI states $\mathcal{QI}$.
  Using the convexity in the same way, for any QI state $\rho^{AB}$ in the form of~\eqref{eq:rho_ab_qip}, we obtain
  \begin{align}
    &\mathcal{E}^{AB}\left(\rho^{AB}\right)\\
    &=\sum_{a,b,j}p\left(j\right)c_{a,b}^{\left(j\right)}\mathcal{E}^{AB}\left(\rho_{a,b}^A\otimes\Ket{j}\Bra{j}^B\right)\in\mathcal{QI},
  \end{align}
  which yields the conclusion.
\end{proof}

As for completely QI-preserving operations, we also show a characterization of completely QI-preserving operations similarly to Proposition~\ref{prp:qi_preserving} on QIP\@.

\begin{proposition}[\label{prp:cqip}Characterization of completely QI-preserving operations]
  A CPTP map $\mathcal{E}^{AB}$ is completely QI-preserving if and only if
  for any $j\in\left\{0,\ldots,D_B-1\right\}$,
  it holds that
  \begin{equation}
    \label{eq:cqip}
    \left(\id^{A^{\prime\prime}}\otimes\mathcal{E}^{AB}\right)\left(\Ket{\Phi_{D_A}}\Bra{\Phi_{D_A}}^{A^{\prime\prime} A}\otimes\Ket{j}\Bra{j}^B\right)\in\mathcal{QI},
  \end{equation}
  where $\mathcal{QI}$ in this proposition means the set of QI states that are quantum on $A^{\prime\prime} A$ and incoherent on $B$, and $\Ket{\Phi_{D_A}}$ is a maximally entangled state of Schmidt rank $D_A$
  \begin{equation}
    \Ket{\Phi_{D_A}}^{A^{\prime\prime} A}\coloneqq\frac{1}{\sqrt{D_A}}\sum_{d=0}^{{D_A}-1}\Ket{d}^{A^{\prime\prime}}\otimes\Ket{d}^A.
  \end{equation}.
\end{proposition}

\begin{proof}
  Any completely QI-preserving map satisfies the condition~\eqref{eq:cqip} by definition, and hence it suffices to show the converse; that is, we prove that for any CPTP map $\mathcal{E}^{AB}$, if $\id^{A^{\prime\prime}}\otimes\mathcal{E}^{AB}$ satisfies~\eqref{eq:cqip}, then $\id^{A^{\prime}B^\prime}\otimes\mathcal{E}^{AB}$ for any auxiliary system $\mathcal{H}^{A^{\prime}}\otimes\mathcal{H}^{B^\prime}$ transforms any QI state into a QI state.
  We write the QI state obtained from~\eqref{eq:cqip} as
  \begin{align}
    &\left(\id^{A^{\prime\prime}}\otimes\mathcal{E}^{AB}\right)\left(\Ket{\Phi_{D_A}}\Bra{\Phi_{D_A}}^{A^{\prime\prime} A}\otimes\Ket{j}\Bra{j}^B\right)\nonumber\\
    &=\sum_{k} q^{\left(j\right)}\left(k\right)\sigma_{k}^{\left(j\right)}{}^{A^{\prime\prime}A}\otimes\Ket{k}\Bra{k}^B.
  \end{align}
  Consider any QI state
  \begin{equation}
    \label{eq:rho_aab_cqip}
    \rho^{A^{\prime}AB^{\prime}B}=\sum_{j,j^\prime} p\left(j,j^\prime\right)\rho_{j,j^\prime}^{A^{\prime}A}\otimes\Ket{j^\prime j}\Bra{j^\prime j}^{B^\prime B},
  \end{equation}
  where we write $\Ket{j^\prime j}=\Ket{j^\prime}\otimes\Ket{j}$.
  Since $\Ket{\Phi_{D_A}}^{A^{\prime\prime} A}$ is a maximally entangled state, for each $j$ and $j^\prime$, there exists a CPTP linear map $\mathcal{E}_{j,j^\prime}^{A^{\prime\prime}\to A^{\prime}}$ from $A^{\prime\prime}$ to $A^{\prime}$ that transforms $\Ket{\Phi_{D_A}}^{A^{\prime\prime} A}$ into $\rho_{j,j^\prime}^{A^{\prime}A}$, that is,
  \begin{equation}
    \left(\mathcal{E}_{j,j^\prime}^{A^{\prime\prime}\to A^{\prime}}\otimes\id^A\right)\left(\Ket{\Phi_{D_A}}\Bra{\Phi_{D_A}}^{A^{\prime\prime} A}\right)=\rho_{j,j^\prime}^{A^{\prime}A}.
  \end{equation}
  Then, we obtain
  \begin{align}
    &\left(\id^{A^{\prime}B^\prime}\otimes\mathcal{E}^{AB}\right)\left(\rho^{A^{\prime}AB^\prime B}\right)\nonumber\\
    &=\sum_{j,j^\prime} p\left(j,j^\prime\right)\left(\id^{A^{\prime}B^\prime}\otimes\mathcal{E}^{AB}\right)\left(\rho_{j,j^\prime}^{A^{\prime}A}\otimes\Ket{j^\prime j}\Bra{j^\prime j}^{B^\prime B}\right)\nonumber\\
    &=\sum_{j,j^\prime} p\left(j,j^\prime\right)\left(\mathcal{E}_{j,j^\prime}^{A^{\prime\prime}\to A^{\prime}}\otimes\id^{B^\prime}\otimes\mathcal{E}^{AB}\right)\left(\Ket{\Phi_{D_A}}\Bra{\Phi_{D_A}}^{A^{\prime\prime} A}\right.\nonumber\\
    &\qquad\left.\otimes\Ket{j^\prime j}\Bra{j^\prime j}^{B^\prime B}\right)\nonumber\\
    &=\sum_{j,j^\prime,k} p\left(j,j^\prime\right)q^{\left(j\right)}\left(k\right)\nonumber\\
    &\qquad\left[\left(\mathcal{E}_{j,j^\prime}^{A^{\prime\prime}\to A^{\prime}}\otimes\id^{A}\right)\left(\sigma_{k}^{\left(j\right)}{}^{A^{\prime\prime}A}\right)\right]\otimes\Ket{j^\prime k}\Bra{j^\prime k}^{B^\prime B}.
  \end{align}
  Since the state in the last line is a QI state, we obtain the conclusion.
\end{proof}

Using the characterization of QIP and CQIP shown in Propositions~\ref{prp:qi_preserving} and~\ref{prp:cqip}, we prove the equivalence of QIP and CQIP as follows.

\begin{theorem}[Equivalence of QI-preserving operations and completely QI-preserving operations]
It holds that
\begin{equation}
\mathsf{CQIP}=\mathsf{QIP}.
\end{equation}
\end{theorem}

\begin{proof}
  It is trivial by definition that any CQIP map is a QIP map, that is,
  \begin{equation}
      \mathsf{CQIP}\subseteqq\mathsf{QIP},
  \end{equation}
  and hence, it suffices to show that any QIP map $\mathcal{E}$ is a CQIP map, that is,
  \begin{equation}
            \mathsf{QIP}\subseteqq\mathsf{CQIP}.
  \end{equation}

  Consider any QIP map $\mathcal{E}^{AB}$ for a bipartite system $A$ and $B$.
  For any $i$, $j$, and $k$, Proposition~\ref{prp:qi_preserving} yields
  \begin{align}
    &\left(\left(\id^A\otimes\Delta^B\right)\circ\mathcal{E}^{AB}\right)\left(\Ket{i}\Bra{j}^A\otimes\Ket{k}\Bra{k}^B\right)\nonumber\\
    &=\mathcal{E}^{AB}\left(\Ket{i}\Bra{j}^A\otimes\Ket{k}\Bra{k}^B\right).
  \end{align}
  Then, it holds that
  \begin{align}
     &\left(\left(\id^{A^{\prime\prime} A}\otimes\Delta^B\right)\circ\left(\id^{A^{\prime\prime}}\otimes\mathcal{E}^{AB}\right)\right)\left(\Ket{i}\Bra{j}^{A^{\prime\prime}}\otimes\Ket{i}\Bra{j}^A\right.\nonumber\\
     &\qquad\left.\otimes\Ket{k}\Bra{k}^B\right)\nonumber\\
     &=\left(\id^{A^{\prime\prime}}\otimes\mathcal{E}^{AB}\right)\left(\Ket{i}\Bra{j}^{A^{\prime\prime}}\otimes\Ket{i}\Bra{j}^A\otimes\Ket{k}\Bra{k}^B\right),
  \end{align}
  where $A^{\prime\prime}$ is a quantum system whose dimension is the same as $A$.
  Due to linearity, we have
  \begin{equation}
    \id^{A^{\prime\prime}}\otimes\mathcal{E}^{AB}\left(\Ket{\Phi_{D_A}}\Bra{\Phi_{D_A}}^{A^{\prime\prime} A}\otimes\Ket{k}\Bra{k}^B\right)\in\mathcal{QI},
  \end{equation}
  where $D_A$ is the dimension of the system $A$, and $\Ket{\Phi_{D_A}}^{A^{\prime\prime} A}$ is a maximally entangled state between $A^{\prime\prime}$ and $A$.
  Therefore, Proposition~\ref{prp:cqip} implies that $\mathcal{E}^{AB}$ is a CQIP map, which yields the conclusion.
\end{proof}

\section{\label{sec:hierarchy}Hierarchy of operations for distributed manipulation of coherence and entanglement}

In this section, we establish hierarchies of operations for distributed manipulation of coherence and entanglement, as depicted in Fig.~\ref{fig:hierarchy}.
In particular, we prove strict inclusion relations among the classes of operations that include LQICC and SQI, or LICC and SI, as special cases.

We prove the following theorem establishing the hierarchies of operations.
The proof of this theorem is based on propositions that we will prove later in this section for showing each strict inclusion in~\eqref{eq:hierarchy_qi} and~\eqref{eq:hierarchy_i} in the theorem.

\begin{theorem}
[\label{thm:inclusion}Hierarchy of operations in distributed manipulation of coherence and entanglement]
  It holds that
  \begin{align}
  \label{eq:hierarchy_qi}
  \begin{split}
    &\textsf{1-LQICC}\subsetneqq\mathsf{LQICC}\subsetneqq\mathsf{SQI}\\
    &\quad\subsetneqq\left(\mathsf{QIP}\cap\mathsf{SEP}\right)\subsetneqq\left(\mathsf{QIP}\cap\mathsf{PPT}\right)\subsetneqq\mathsf{QIP},
  \end{split}\\
  \label{eq:hierarchy_i}
  \begin{split}
      &\textsf{1-LICC}\subsetneqq\mathsf{LICC}\subsetneqq\mathsf{SI}\\
      &\quad\subsetneqq\left(\mathsf{MIO}\cap\mathsf{SEP}\right)\subsetneqq\left(\mathsf{MIO}\cap\mathsf{PPT}\right)\subsetneqq\mathsf{MIO},
  \end{split}
  \end{align}
  where $\mathsf{MIO}$ is the set of maximally incoherent operations on $AB$.
\end{theorem}

\begin{proof}
  The inclusions $\textsf{1-LQICC}\subsetneqq\mathsf{LQICC}$ and  $\textsf{1-LICC}\subsetneqq\mathsf{LICC}$ follow from Proposition~\ref{prp:r_licc} using one-shot entanglement distillation.
  Reference~\cite{S1} shows $\mathsf{LQICC}\subsetneqq\mathsf{SQI}$ and $\mathsf{LICC}\subsetneqq\mathsf{SI}$ using local state discrimination.
  We show $\mathsf{SQI}\subsetneqq\left(\mathsf{QIP}\cap\mathsf{SEP}\right)$ and $\mathsf{SI}\subsetneqq\left(\mathsf{MIO}\cap\mathsf{SEP}\right)$ in Proposition~\ref{prp:sep} using coherence dilution.
  As for $\left(\mathsf{QIP}\cap\mathsf{SEP}\right)\subsetneqq\left(\mathsf{QIP}\cap\mathsf{PPT}\right)$ and $\left(\mathsf{MIO}\cap\mathsf{SEP}\right)\subsetneqq\left(\mathsf{MIO}\cap\mathsf{PPT}\right)$, we prove the separation in Proposition~\ref{prp:sep} using local state discrimination, whereas a conventional way of separating $\mathsf{SEP}$ and $\mathsf{PPT}$ by considering preparation of a PPT entangled state is insufficient.
  We prove $\left(\mathsf{MIO}\cap\mathsf{PPT}\right)\subsetneqq\mathsf{MIO}$ by showing an instance of \textsf{MIO} that is not in $\mathsf{MIO}\cap\mathsf{PPT}$, which is the SWAP operation.
  Note that the SWAP operation is in \textsf{MIO} but not in \textsf{QIP}, and hence it is not straightforward to generalize this example to $\left(\mathsf{QIP}\cap\mathsf{PPT}\right)\subsetneqq\mathsf{QIP}$.
  To show $\left(\mathsf{QIP}\cap\mathsf{PPT}\right)\subsetneqq\mathsf{QIP}$, we will investigate assisted distillation of coherence in the next section, especially in Proposition~\ref{prp:fidelity}. 
\end{proof}

In the following, we show propositions used in the above proof of the hierarchies.
The following proposition shows that the difference between $r$-round LQICC/LICC and $(r-1)$-round LQICC/LICC for any $r\geqq 1$, respectively, arises in one-shot entanglement distillation.

\begin{proposition}[\label{prp:r_licc}Separation between $r$-round LQICC/LICC and $(r-1)$-round LQICC/LICC] For any $r\in\left\{1,2,\ldots\right\}$,
there exists a quantum state $\rho_r^{AB}$ shared between $A$ and $B$ for which one-shot entanglement distillation of $1$ ebit is achievable by $r$-round LICC but not achievable by $(r-1)$-round LOCC, leading to a separation between $r$- and $(r-1)$-round LICC, and a separation between $r$- and $(r-1)$-round LQICC\@. 
\end{proposition}

\begin{proof}
  Consider a family of quantum states that are represented as a convex combination of maximally entangled states corresponding to the origami distribution in Ref.~\cite{E1}, where each of the maximally entangled states is represented in terms of the reference basis for $A$ and $B$.
  Then, Ref.~\cite{E1} shows that for each $r\in\left\{1,2,\ldots\right\}$, there is a mixed state in this family for which one-shot entanglement distillation of $1$ ebit cannot be achieved by any $(r-1)$-round LOCC, yet there exists an $r$-round LOCC protocol that achieves one-shot entanglement distillation of $1$ ebit from the same mixed state.
  Since $(r-1)$-round LOCC includes $(r-1)$-round LICC and  $(r-1)$-round LQICC\@, this one-shot entanglement distillation cannot be achieved by any $(r-1)$-round LICC or any $(r-1)$-round LQICC\@.

  If each maximally entangled state for this mixed state is represented in terms of the reference basis for $A$ and $B$, the $r$-round LOCC protocol shown in Ref.~\cite{E1} consists only of incoherent measurements; that is, this protocol is an $r$-round LICC protocol, which yields the separation between $r$-round LICC and $(r-1)$-round LICC\@. 
  As for the separation between $r$-round LQICC and $(r-1)$-round LQICC, since $r$-round LICC is a subset of $r$-round LQICC\@, the above $r$-round LICC protocol also yields the separation between $r$-round LQICC and $(r-1)$-round LQICC\@.
\end{proof}

As for the separation between $\mathsf{SQI}$ and $\mathsf{QIP}\cap\mathsf{SEP}$ and that between $\mathsf{SI}$ and $\mathsf{MIO}\cap\mathsf{SEP}$, we can consider a special case where $A$ is one-dimensional, so that these operations can be regarded as \textsf{IO} and \textsf{MIO} on $B$.
Then, the difference between \textsf{IO} and \textsf{MIO} yields the following separation. 

\begin{proposition}
[\label{prp:sep}Separations between $\mathsf{SQI}$ and $\mathsf{QIP}\cap\mathsf{SEP}$ and between $\mathsf{SI}$ and $\mathsf{MIO}\cap\mathsf{SEP}$]
It holds that
\begin{align}
    \label{eq:sqi_qip_sep}
    \mathsf{SQI}&\subsetneqq\mathsf{QIP}\cap\mathsf{SEP},\\
    \label{eq:si_ip_sep}
    \mathsf{SI}&\subsetneqq\mathsf{MIO}\cap\mathsf{SEP}.
\end{align}
\end{proposition}

\begin{proof}
The separation in~\eqref{eq:sqi_qip_sep} and~\eqref{eq:si_ip_sep} is a special case of separation between IO and MIO, by considering $A$'s system to be one-dimensional.
This separation between IO and MIO can be seen in tasks such as coherence dilution~\cite{Z1}.
Therefore, when $A$'s system is one-dimensional, there exists an MIO on $B$ included in $\mathsf{QIP}\cap\mathsf{SEP}\setminus\mathsf{SQI}$ and $\mathsf{MIO}\cap\mathsf{SEP}\setminus\mathsf{SI}$, whereas any IO on $B$ is included in $\textsf{LQICC}$ and $\textsf{LICC}$, and hence in $\textsf{SQI}$ and $\textsf{SI}$.
\end{proof}

To show separation between $\mathsf{QIP}\cap\mathsf{SEP}$ and $\mathsf{QIP}\cap\mathsf{PPT}$, and show that between $\mathsf{MIO}\cap\mathsf{SEP}$ and $\mathsf{MIO}\cap\mathsf{PPT}$,
we use local state discrimination.
Note that the separation between $\mathsf{SEP}$ and $\mathsf{PPT}$ in the entanglement theory is first proven in Ref.~\cite{H1} based on the fact that there exists a free state of $\mathsf{PPT}$, a PPT state, that is not a free state of $\mathsf{SEP}$, a separable state.
However, this proof based on free states in the entanglement theory is not applicable to our case because the sets of free states for $\mathsf{QIP}\cap\mathsf{PPT}$ and $\mathsf{MIO}\cap\mathsf{PPT}$ are included in $\mathcal{QI}$ and hence are always separable.
Rather than this conventional argument based a PPT entangled state, our proof generalizes a more recent result on separating \textsf{SEP} and \textsf{PPT} based on local state discrimination as follows. 

\begin{proposition}[\label{prp:qip_sep_mio_sep}Separations between $\mathsf{QIP}\cap\mathsf{SEP}$ and $\mathsf{QIP}\cap\mathsf{PPT}$ and between $\mathsf{MIO}\cap\mathsf{SEP}$ and $\mathsf{MIO}\cap\mathsf{PPT}$]
It holds that
\begin{align}
    \label{eq:qip_sep_qip_ppt}
    \left(\mathsf{QIP}\cap\mathsf{SEP}\right)&\subsetneqq\left(\mathsf{QIP}\cap\mathsf{PPT}\right),\\
    \label{eq:mio_sep_mio_ppt}
    \left(\mathsf{MIO}\cap\mathsf{SEP}\right)&\subsetneqq\left(\mathsf{MIO}\cap\mathsf{PPT}\right).
\end{align}
\end{proposition}

\begin{proof}
  Recall that there exists a set of mutually orthogonal bipartite pure states for which there exists a PPT measurement achieving success probability $7/8$ in state discrimination~\cite{C4}, and hence a PPT map implemented by this PPT measurement can achieve this success probability.
  However, it is also known that no separable measurement can achieve success probability greater than $3/4$ in this state discrimination~\cite{B1}, and hence by definition, no separable map can achieve this success probability.
  Note that these PPT and separable measurements can be regarded as CPTP maps that transform any input state into an incoherent state representing the probabilistic mixture of the measurement outcomes.
  These CPTP maps corresponding to the measurements are in $\mathsf{QIP}$ and $\mathsf{MIO}$ because the output states are always incoherent.
  Thus, this local state discrimination yields the separations of~\eqref{eq:qip_sep_qip_ppt} and~\eqref{eq:mio_sep_mio_ppt} in terms of the success probability.
\end{proof}

We prove the separation between $\mathsf{MIO}\cap\mathsf{PPT}$ and \textsf{MIO} by showing an instance.
In general, it is not straightforward to identify an instance of operation that is in one class but is probably not in the other class, because the latter no-go theorem is generally hard to prove.
However, we here show that the SWAP operation is such an instance in this case.

\begin{proposition}[\label{prp:mio_ppt}Separation between $\mathsf{MIO}\cap\mathsf{PPT}$ and \textsf{MIO}]
Given any bipartite system $AB$, it holds that
\begin{equation}
    \left(\mathsf{MIO}\cap\mathsf{PPT}\right)\subsetneqq\mathsf{MIO},
\end{equation}
where $\mathsf{MIO}$ is the set of maximally incoherent operations on $AB$.
\end{proposition}

\begin{proof}
  A SWAP unitary operation is an MIO map because SWAP transforms any incoherent state into a swapped incoherent state,
  but this SWAP is not a PPT map because the partial transpose of the Choi operator of SWAP is not positive semidefinite.
  Hence, this SWAP is an example showing the conclusion.
\end{proof}

\section{\label{sec:non_surviving_separation}Asymptotically non-surviving separation of hierarchy in assisted distillation of coherence}

In this section, we investigate a phenomena of asymptotically non-surviving separation~\cite{Y1} of the hierarchy, especially between $1$-\textsf{LQICC} and \textsf{QIP}, in achieving an information-processing task, specifically, assisted distillation of coherence.
The assisted distillation of coherence involves two parties $A$ and $B$ sharing some initial state, and $A$ can perform an arbitrary quantum operation while $B$ has a restriction in generating coherence. The aim of the task is to distill on $B$ as much coherence as possible from the shared initial state with assistance of $A$.
While we prove difference between the classes of operations for distributed manipulation of coherence in the previous section, this difference does not necessarily affect achievability of certain tasks.
Indeed, while $1$-\textsf{LQICC} and \textsf{SQI} are different as the sets of operations, it is known that this difference does not appear in the task of asymptotic assisted distillation of coherence for any pure state~\cite{PhysRevLett.116.070402}.
We here generalize this result to show that even the difference between $1$-\textsf{LQICC} and \textsf{QIP} does not appear in this asymptotic scenario; that is, all the operations in this hierarchy have the same power in achieving this asymptotic task for any initial pure state.
At the same time, we consider a one-shot scenario of this task and discover an instance of the initial pure state for which the amount of distillable coherence with assistance is different under $\mathsf{QIP}\cap\mathsf{PPT}$ and \textsf{QIP}; that is, the separation between $1$-\textsf{LQICC} and \textsf{QIP} arises in this one-shot scenario.
While the proof of no-go theorem for showing the separation in one-shot scenarios is hard in general, we identify such an instance of separation by exploiting a numerical algorithm by SDP based on the characterization of \textsf{QIP} that we have proven in Proposition~\ref{prp:qi_preserving}.

The task of assisted distillation of coherence~\cite{PhysRevLett.116.070402,M1,R1} is defined as follows.
Consider two parties $A$ and $B$.
While the original formulation~\cite{PhysRevLett.116.070402} considers $1$-\textsf{LQICC} as free operations performed by $A$ and $B$, we here generalize the situation to other classes of operations in the hierarchy, and let $\mathcal{O}$ denote a class of free operations on $A$ and $B$ that is in the hierarchy~\eqref{eq:hierarchy_qi} in Theorem~\ref{thm:inclusion}, such as \textsf{QIP}.
Given an arbitrary mixed state $\psi^B$ on $B$, we consider its purification $\Ket{\psi}^{AB}$ shared initially between $A$ and $B$, where $\psi^B=\tr_A\Ket{\psi}\Bra{\psi}^{AB}$.
The state to be distilled on $B$ is a maximally coherent state denoted by
\begin{equation}
  \Ket{\Phi_M}^B\coloneqq\frac{1}{\sqrt{M}}\sum_{j=0}^{M-1}\Ket{j}^B,
\end{equation}
where ${\left\{\Ket{j}\right\}}_j$ is the fixed reference basis on $B$, and $M$ may be called the coherence rank of $\Ket{\Phi_M}^B$, which quantifies the coherence.
To distill as much coherence as possible, the parties perform the free operations $\mathcal{O}$ to transform the initial state, \textit{i.e.}, $\Ket{\psi}^{AB}$, to a maximally coherent state on $B$ up to a given error $\epsilon$ in the trace distance, so that $M$ of $B$'s final maximally coherent state can be maximized.
In a one-shot scenario, the one-shot distillable coherence with assistance is defined as
\begin{align}
    &C_{\mathrm{a},1,\epsilon}^{\mathcal{O}}\left(\psi^B\right)\coloneqq\max\left\{\log_2 M:\right.\nonumber\\
    &\quad \left\|\mathcal{E}^{AB}\left(\Ket{\psi}\Bra{\psi}^{AB}\right)-\Ket{\Phi_M}\Bra{\Phi_M}^B\right\|_1\leqq \epsilon,\nonumber\\
    &\quad M\in\mathbb{N},\,\mathcal{E}^{AB}\in\mathcal{O}\left.\right\}.
\end{align}
Correspondingly in the asymptotic scenario where the task is repeated infinitely many times within a vanishing error, the asymptotic distillable coherence with assistance is defined as
\begin{equation}
  C_{\mathrm{a},\infty}^{\mathcal{O}}\left(\psi^B\right)\coloneqq\lim_{\epsilon\to 0}\lim_{n\to\infty}\frac{1}{n}C_{\mathrm{a},1,\epsilon}^{\mathcal{O}}\left({\left(\psi^B\right)}^{\otimes n}\right),
\end{equation}

We show that the asymptotic distillable coherence with assistance under \textsf{QIP} is always equal to that under  $1$-\textsf{LQICC} for any pure initial state $\Ket{\psi}^{AB}$.
In other words, all the operations in the hierarchy~\eqref{eq:hierarchy_qi} in Theorem~\ref{thm:inclusion} have the same power in this task.

\begin{proposition}[\label{prp:asymptotic}Equivalence of asymptotic distillable coherence with assistance under \textsf{QIP} and $1$-\textsf{LQICC}]
  Given any bipartite pure state $\Ket{\psi}^{AB}$,
  it holds that
  \begin{equation}
    C_{\mathrm{a},\infty}^\textsf{QIP}\left(\psi^{B}\right)=C_{\mathrm{a},\infty}^\textsf{1-LQICC}\left(\psi^{B}\right)=H\left(\Delta\left(\psi^B\right)\right),
  \end{equation}
  where $\psi^B=\tr_A\Ket{\psi}\Bra{\psi}^{AB}$ is the reduced state of $B$ for $\Ket{\psi}^{AB}$, $\Delta$ is the completely dephasing channel, and $H$ is the quantum entropy given by
  \begin{equation}
      H\left(\psi\right)\coloneqq -\tr\psi\log_2\psi.
  \end{equation}
\end{proposition}

\begin{proof}
  It suffices to prove
  \begin{equation}
    \label{eq:qip_c_bound}
    C_{\mathrm{a},\infty}^\textsf{QIP}\left(\psi^{B}\right)\leqq H\left(\Delta\left(\psi^B\right)\right),
  \end{equation}
  since Theorem~\ref{thm:inclusion} yields
  \begin{equation}
    C_{\mathrm{a},\infty}^\textsf{1-LQICC}\left(\psi^{B}\right)\leqq C_{\mathrm{a},\infty}^\textsf{QIP}\left(\psi^{B}\right),
  \end{equation}
  and it is shown in Ref.~\cite{PhysRevLett.116.070402} that
  \begin{equation}
    C_{\mathrm{a},\infty}^\textsf{1-LQICC}\left(\psi^{B}\right)\geqq H\left(\Delta\left(\psi^B\right)\right).
  \end{equation}

  To prove Inequality~\eqref{eq:qip_c_bound}, we use the QI relative entropy~\cite{PhysRevLett.116.070402}
  \begin{equation}
    C_r^{A|B}\left(\rho^{AB}\right)\coloneqq\min_{\sigma^{AB}\in\mathcal{QI}}D\left(\rho^{AB}||\sigma^{AB}\right),
  \end{equation}
  where $D\left(\rho||\sigma\right)\coloneqq\tr\rho\log_2\rho-\tr\rho\log_2\sigma$ is the quantum relative entropy.
  The QI relative entropy is monotonically nonincreasing under QI-preserving maps because for any $\mathcal{E}^{AB}\in\mathsf{QIP}$ and $\rho^{AB}$, we have
  \begin{align}
    &C_r^{A|B}\left(\rho^{AB}\right)\\
    \label{eq:qi_relative}
    &=\min_{\sigma^{AB}\in\mathcal{QI}}D\left(\rho^{AB}||\sigma^{AB}\right)\\
    \label{eq:min}
    &= D\left(\rho^{AB}||\sigma_{\min}^{AB}\right)\\
    &\geqq D\left(\mathcal{E}^{AB}\left(\rho^{AB}\right)||\mathcal{E}^{AB}\left(\sigma_{\min}^{AB}\right)\right)\\
    &\geqq\min_{\sigma^{AB}\in\mathcal{QI}}D\left(\mathcal{E}\left(\rho^{AB}\right)||\sigma^{AB}\right)\\
    &=C_r^{A|B}\left(\mathcal{E}^{AB}\left(\rho^{AB}\right)\right),
  \end{align}
  where $\sigma_{\min}^{AB}\in\mathcal{QI}$ in~\eqref{eq:min} is a state achieving the minimum in~\eqref{eq:qi_relative}.

  Consider any QI-preserving map $\mathcal{E}^{AB\to B}$ that achieves assisted distillation of coherence as follows
  \begin{equation}
    {\left\|\mathcal{E}^{AB\to B}\left({\Ket{\psi}\Bra{\psi}^{\otimes n}}^{AB}\right)-\Ket{\Phi_M}\Bra{\Phi_M}^B\right\|}_1\leqq\epsilon,
  \end{equation}
  where the output of $\mathcal{E}^{AB\to B}$ on $A$ is traced out.
  Reference~\cite{PhysRevLett.116.070402} also shows the continuity of the QI relative entropy in the sense that
  for any $\rho^{B}$ and $\sigma^{B}$ where
  \begin{equation}
    T\coloneqq{\left\|\rho^{B}-\sigma^{B}\right\|}_1<1,
  \end{equation}
  it holds that
  \begin{equation}
    \left|C_r^{A|B}\left(\rho^{B}\right)-C_r^{A|B}\left(\sigma^{B}\right)\right|\leqq 2T\log_2\dim\mathcal{H}^B+2h\left(T\right),
  \end{equation}
  where $\mathcal{H}^B$ is the Hilbert space for $\rho^B$ and $\sigma^B$, and $h\left(x\right)\coloneqq -x\log_2 x-\left(1-x\right)\log_2\left(1-x\right)$ is the binary entropy.
  Using this continuity in the case of $0<\epsilon\leqq\frac{1}{2}$,
  we obtain
  \begin{equation}
    \begin{split}
      &C_r^{A|B}\left(\mathcal{E}^{AB\to B}\left({\Ket{\psi}\Bra{\psi}^{\otimes n}}^{AB}\right)\right)\\
      &\geqq C_r^{A|B}\left(\Ket{\Phi_M}\Bra{\Phi_M}^B\right)-2n\epsilon\log_2 d-2h\left(\epsilon\right)\\
    \end{split}
  \end{equation}
  where $d^n$ is the dimension of the system $B$.
  Since we have~\cite{PhysRevLett.116.070402}
  \begin{align}
    C_r^{A|B}\left(\Ket{\psi}\Bra{\psi}^{AB}\right)&=H\left(\Delta\left(\psi^B\right)\right),\\
    C_r^{A|B}\left(\Ket{\Phi_M}\Bra{\Phi_M}^{B}\right)&=\log_2 M,
  \end{align}
  we obtain
  \begin{align}
    &nH\left(\Delta\left(\psi^B\right)\right)\\
    &\geqq C_r^{A|B}\left(\mathcal{E}^{AB\to B}\left({\left(\Ket{\psi}\Bra{\psi}^{\otimes n}\right)}^{AB}\right)\right)\\
    &\geqq\log_2 M-2n\epsilon\log_2 d-2h\left(\epsilon\right),
  \end{align}
  Thus, for any $\epsilon$ satisfying $0<\epsilon\leqq\frac{1}{2}$, it is necessary that the coherence of assistance under any QI-preserving map should satisfy
  \begin{equation}
    \begin{split}
    &\frac{\log_2 M}{n}=H\left(\Delta\left(\psi^B\right)\right)+2\epsilon\log_2 d + O\left(\frac{1}{n}\right)\\
    &\text{as }n\to\infty,
    \end{split}
  \end{equation}
  which yields Inequality~\eqref{eq:qip_c_bound}.
\end{proof}

In contrast with this coincidence $C_{\mathrm{a},\infty}^{\textsf{1-LQICC}}\left(\psi^{B}\right)=C_{\mathrm{a},\infty}^{\textsf{QIP}}\left(\psi^{B}\right)$ in the asymptotic scenario for any pure state $\Ket{\psi}^{AB}$, we here identify an instance of $\Ket{\psi}^{AB}$ in the one-shot scenario that shows the separation $C_{\mathrm{a},1,\epsilon}^{\textsf{1-LQICC}}\left(\psi^{AB}\right)<C_{\mathrm{a},1,\epsilon}^{\textsf{QIP}}\left(\psi^{AB}\right)$ for some $\epsilon$.
To identify such an instance, we need a tight upper bound of  $C_{\mathrm{a},1,\epsilon}^{\textsf{1-LQICC}}\left(\psi^{AB}\right)$ and at the same time, a protocol by \textsf{QIP} that can outperform this upper bound.
While some general upper bounds of $C_{\mathrm{a},1,\epsilon}^{\textsf{1-LQICC}}\left(\psi^{AB}\right)$ are given in Refs.~\cite{M1,R1}, the difficulty in showing the instance arises from the fact that these bounds are not tight enough to provide an optimal bound under $1$-\textsf{LQICC}, and hence it is unclear whether one can construct a \textsf{QIP} protocol outperforming this bound. 
Rather than the general bounds, we here need a tighter bound for showing a particular example.

One possible candidate of the instance is a counterexample of the tightness of the general upper bounds shown in Refs.~\cite{M1,R1}.
The counterexample exists when $D_B=\dim\mathcal{H}^B=4$, as shown in the following.
Choose $t\in\left(0,\frac{1}{2}\right)$ and $\omega\coloneqq\frac{1}{2}\left(-1+\textup{i}\sqrt{3}\right)$.
Define
\begin{align}
\Ket{u}&\coloneqq\frac{1}{\sqrt{3}}\left(\begin{array}{c}
    \omega\\
    \omega^2\\
    1\\
    0
\end{array}\right),\\
\Ket{v\left(t\right)}&\coloneqq\left(\begin{array}{c}
        t\\
        t\\
        t\\
        \sqrt{1-3t^2}
    \end{array}\right).
\end{align}
Consider a bipartite pure state
\begin{equation}
\label{eq:example}
\begin{split}
  &\Ket{\psi}^{AB}\coloneqq\\
  &\sqrt{1-\frac{1}{4-12t^2}}\Ket{0}^A\otimes\Ket{u}^B+\sqrt{\frac{1}{4-12t^2}}\Ket{1}^A\otimes\Ket{v\left(t\right)}^B.
\end{split}
\end{equation}
It is shown in Refs.~\cite{R1,C3} that this state cannot be transformed into the maximally coherent state of coherence rank $4$ by any one-way LQICC, while the general upper bound of one-shot assisted coherence with assistance is as loose as $\log_2 4=2$ as $\epsilon\to 0$.
This no-go theorem is based on the fact that $\psi^B$ is a rank-$2$ extremal point of positive semidefinite operators  on $\mathbb{C}^4$ satisfying $\Delta\left(\psi^B\right)=\frac{\openone}{4}$, as shown in Example~4 in Ref.~\cite{C3}.
However, it has been unknown whether there exists a protocol for assisted distillation of coherence that uses a more powerful class of operations than one-way LQICC to outperform the bound of this no-go theorem for this choice of $\Ket{\psi}^{AB}$.

For the state~\eqref{eq:example},
we prove the following separation of one-shot distillable coherence with assistance, using SDP-based numerical calculation as a tool for evaluating the bound for the separation.
Remark that this separation does not survive in the asymptotic scenario as shown in Proposition~\ref{prp:asymptotic}.

\begin{proposition}
[\label{prp:fidelity}Separation between $\mathsf{QIP}\cap\mathsf{PPT}$ and \textsf{QIP} in one-shot assisted distillation of coherence]
  For the quantum state $\Ket{\psi}^{AB}$ defined as~\eqref{eq:example} where we set the parameter $t=1/4$,
  it holds that
  \begin{align}
  \label{eq:qip_bound}
    \max_{\mathcal{E}\in\mathsf{QIP}}F^2\left(\Ket{\Phi_4}\Bra{\Phi_4}^B,\mathcal{E}^{AB\to B}\left(\Ket{\psi}\Bra{\psi}^{AB}\right)\right)&=1,\\
  \label{eq:qip_ppt_bound}
    \max_{\mathcal{E}\in\mathsf{QIP}\cap\mathsf{PPT}}F^2\left(\Ket{\Phi_4}\Bra{\Phi_4}^B,\mathcal{E}^{AB\to B}\left(\Ket{\psi}\Bra{\psi}^{AB}\right)\right)&<0.98,
  \end{align}
  where $\Ket{\Phi_4}$ is the maximally coherent state of coherence rank $4$, and $F^2\left(\rho,\sigma\right)\coloneqq\left\|\sqrt{\rho}\sqrt{\sigma}\right\|_1^2$ is the fidelity. As a result, for some nonzero error $\epsilon>0$, there exists a separation between $\mathsf{QIP}\cap\mathsf{PPT}$ and \textsf{QIP} in one-shot distillable coherence with assistance 
  \begin{equation}
  \label{eq:separation_qip_ppt}
      C_{\mathrm{a},1,\epsilon}^{\mathsf{QIP}\cap\mathsf{PPT}}\left(\psi^{B}\right)<C_{\mathrm{a},1,\epsilon}^{\mathsf{QIP}}\left(\psi^{B}\right)=2.
  \end{equation}
\end{proposition}

\begin{proof}
We calculate~\eqref{eq:qip_bound} and~\eqref{eq:qip_ppt_bound} by SDP\@.
To calculate~\eqref{eq:qip_bound}, based on Proposition~\ref{prp:qi_preserving},  maximize
\begin{equation}
    \tr\left[\left(\Ket{\Phi_4}\Bra{\Phi_4}^{B^\prime}\otimes{\left(\Ket{\psi}\Bra{\psi}^{A B}\right)}^\mathrm{T}\right){J\left(\mathcal{E}\right)}^{B^\prime A B}\right]
\end{equation}
subject to
\begin{align}
\label{eq:CP}
&{J\left(\mathcal{E}\right)}^{B^\prime A B}\geqq 0,\\
\label{eq:TP}
&\tr_{B^\prime}{J\left(\mathcal{E}\right)}^{B^\prime A B}=\openone^{AB},\\
\label{eq:def_sigma}
&\sigma_{a,b,j}^{B^\prime}=\tr_{AB}\left[\left(\openone^{B^\prime}\otimes{\left(\rho_{a,b}^A\otimes\Ket{j}\Bra{j}^B\right)}^\mathrm{T}\right){J\left(\mathcal{E}\right)}^{B^\prime A B}\right],\\
\label{eq:condition_sigma}
&\Delta\left(\sigma_{a,b,j}^{B^\prime}\right)=\sigma_{a,b,j}^{B^\prime},\,\forall a,b,j.
\end{align}
where ${\left(\Ket{\psi}\Bra{\psi}^{A B}\right)}^\mathrm{T}$ is the transpose of $\Ket{\psi}\Bra{\psi}^{A B}$ with respect to the reference basis, $\Ket{\Phi_4}$ is the maximally coherent state of coherence rank $4$, $J\left(\mathcal{E}\right)$ is the Choi operator~\cite{Watrous:2018:TQI:3240076} of $\mathcal{E}$, $\openone$ is the identity operator, and $\Delta$ is the completely dephasing channel.
Note that~\eqref{eq:CP} and~\eqref{eq:TP} represent the completely positive and trace-preserving properties respectively~\cite{Watrous:2018:TQI:3240076},~\eqref{eq:def_sigma} is obtained from Proposition~\ref{prp:qi_preserving}, and~\eqref{eq:condition_sigma} is by definition of incoherent states. 
As for~\eqref{eq:qip_ppt_bound}, in the same say,  maximize
\begin{equation}
    \tr\left[\left(\Ket{\Phi_4}\Bra{\Phi_4}^{B^\prime}\otimes{\left(\Ket{\psi}\Bra{\psi}^{A B}\right)}^\mathrm{T}\right){J\left(\mathcal{E}\right)}^{B^\prime A B}\right]
\end{equation}
subject to
\begin{align}
&{J\left(\mathcal{E}\right)}^{B^\prime A B}\geqq 0,\\
&\tr_{B^\prime}{J\left(\mathcal{E}\right)}^{B^\prime A B}=\openone^{AB},\\
&\sigma_{a,b,j}^{B^\prime}=\tr_{AB}\left[\left(\openone^{B^\prime}\otimes{\left(\rho_{a,b}^A\otimes\Ket{j}\Bra{j}^B\right)}^\mathrm{T}\right){J\left(\mathcal{E}\right)}^{B^\prime A B}\right],\\
&\Delta\left(\sigma_{a,b,j}^{B^\prime}\right)=\sigma^{B^\prime},\,\forall a,b,j\\
\label{eq:ppt_condition}
&{\left({J\left(\mathcal{E}\right)}^{B^\prime A B}\right)}^{\mathrm{T}_A}\geqq 0,,
\end{align}
where ${\left({J\left(\mathcal{E}\right)}^{B^\prime A B}\right)}^{\mathrm{T}_A}$ is the partial transpose of ${J\left(\mathcal{E}\right)}^{B^\prime A B}$ on $A$ with respect to the reference basis.
Note that~\eqref{eq:ppt_condition} represents the condition of PPT maps~\cite{Rains2000}.
Since these optimizations are SDP~\cite{Watrous:2018:TQI:3240076}, we numerically calculate~\eqref{eq:qip_bound} and~\eqref{eq:qip_ppt_bound} using YALMIP~\cite{L} and Splitting Conic Solver (SCS)~\cite{O} with the sufficient precision to show the separation, where the precision of the output solution compared to the optimal solution is checked in the numerical algorithm by the duality of SDP\@.
\end{proof}

\section{\label{sec:conclusion}Conclusion}

We have investigated the power of different classes of operations in the manipulation of coherence and entanglement in distributed settings, in particular, in a quantum-incoherent (QI) setting and an incoherent-incoherent setting as illustrated in Fig.~\ref{fig:introduction}.
In the quantum-incoherent setting,
we consider and analyze QI-preserving maps \textsf{QIP}, that is, the maximal set of free operations.
Unlike the entanglement theory where separability-preserving maps \textsf{SEPP} and separable maps \textsf{SEP} (\textit{i.e.}, completely \textsf{SEPP}) are different, we prove that the corresponding maps in the distributed manipulation of coherence and entanglement, that is, \textsf{QIP} and the completely QI-preserving maps \textsf{CQIP}, are the same set of maps.
As for the incoherent-incoherent setting, maximally incoherent operations \textsf{MIO} serve as the maximal set of free operations.
In contrast with previously known classes of operations in these settings that have the inclusion relations shown in~\eqref{eq:inclusion_relation},
we show that no inclusion relation holds between \textsf{QIP} and \textsf{MIO}.
These results highlight the difference between the entanglement theory and the resource theory of distributed manipulation of coherence and entanglement investigated in this paper.

Using these operations, we establish a hierarchy of the classes of free operations in the quantum-incoherent setting and that in the incoherent-incoherent setting.
We introduce different classes of operations by starting with combinations of the smallest sets of maps (\textsf{LOCC} and \textsf{IO}), building up to larger ones by moving up the entanglement and coherence hierarchy (to \textsf{PPT} and \textsf{MIO}), and considering the set intersections.
We show the separation among these operations for distributed manipulation of coherence and entanglement as shown in Fig.~\ref{fig:hierarchy}.
In contrast with previously known classes of operations, some of the operations in Fig.~\ref{fig:hierarchy} that we have introduced can be used for numerical algorithms based on semidefinite programming (SDP).
Assisted by SDP, we have discovered that in the task of one-shot assisted distillation of coherence, the hierarchy collapses in terms of asymptotic distillable coherence of assistance, but there still exists a non-zero separation between \textsf{QIP} and $\mathsf{QIP}\cap\mathsf{PPT}$ in the corresponding one-shot task.
Our results clarify tasks where the differences of these operations for manipulating coherence and entanglement appear, and at the same time demonstrate the task where the asymptotically non-surviving separation of the hierarchy arises.

Finally, towards further understandings of the structure of distributed manipulation of  coherence and entanglement, we pose some questions that we believe are interesting to investigate in future research.
While we have discussed the subtleties of taking intersection of two classes of operations to introduce another class of operations in Remark~\ref{remark:op_int}, a natural question to ask in this context is whether the operational intersection of $\mathsf{LOCC}$ and $\mathsf{IO}$ is the same as $ \mathsf{LICC}$ or not, which is a refinement of an open question raised in Ref.~\cite{S2}. 
Apart from the operations in the hierarchy in Fig.~\ref{fig:hierarchy}, it is interesting to see whether we can define other physically motivated or numerically tractable classes of operations that can be situated in the hierarchies, especially without taking intersection of existing operations.
Regarding the asymptotically non-surviving separation of the hierarchy, it is still open whether we can demonstrate an asymptotically non-surviving separation between more experimentally tractable classes of operations than QIP, especially $r$-round LQICC/LICC and $(r+1)$-round LQICC/LICC for some $r$ in a certain task.
Our results provide fundamental techniques for tackling these types of questions, opening the way to further understandings of interplay between quantum coherence and quantum entanglement as quantum resources.

\acknowledgments{This work was supported by CREST (Japan Science and Technology Agency) JPMJCR1671 and Cross-ministerial Strategic Innovation Promotion Program (SIP) (Council for Science, Technology and Innovation (CSTI)). We acknowledge Yoshifumi Nakata for holding a workshop and inviting us to work together.}

\bibliographystyle{plainnat}
\bibliography{citation}

\end{document}